\documentclass[conference,letterpaper]{IEEEtran}

\usepackage{amsthm}
\usepackage{amsmath} %
\usepackage{mathtools}
\usepackage{graphicx}
\usepackage{psfrag}
\usepackage{epsf,epsfig}
\usepackage{epstopdf}
\usepackage{subfigure}
\usepackage{cite,color}
\usepackage{cite}
\usepackage{amsmath,amssymb,amsfonts}
\usepackage{algorithmic}
\usepackage{graphicx}
\usepackage{textcomp}
\usepackage[x11names]{xcolor}
\usepackage{verbatim}
\usepackage{amsmath}

\definecolor{red}{rgb}{1,0,0}  

\newtheorem{Prop}{Proposition}

\usepackage{graphicx}
\usepackage{tikz}
\usepackage{caption}

\definecolor{MyRoyalBlue}{RGB}{65,105,225}
\definecolor{MyFireBrick}{RGB}{178,34,34}

\begin{document}

\title{Signal Constellation Construction \\ via Radio Frequency Mirrors }

\vspace{-.3cm}
\author{
\IEEEauthorblockN{Majid Nasiri Khormuji} \\
\vspace{-.5cm}
\IEEEauthorblockA{{Huawei Technologies Sweden AB}\\
                           majid.n.k@ieee.org
                          \vspace{-.6cm}
                           }

}

\maketitle

\begin{abstract}
By integrating feedback with Radio Frequency (RF) mirrors, we develop a closed-loop media-based modulation system for efficient utilization of the signal space. Specifically, this closed-loop construction optimizes the inherited signal constellation from the media, achieving a significantly larger minimum pairwise Euclidean distance than the original configuration. The initial signal constellation, derived from the media, is used to compute a set of complex weights for all activation patterns of the RF mirrors. These complex weights are then fed back to the transmitter to refine the transmit signal before it reaches the mirrors. This feedback mechanism ensures that the received, shaped signal constellation retains improved properties, enabling more reliable transmission. Notably, the closed-loop approach enables the media-based modulation to approach the performance of an AWGN channel, while the channel from each mirror to the single-antenna receiver is modeled as Rayleigh fading.
\end{abstract}


\section{Introduction}\label{sec:intro}

Recent innovations in wireless communication have underscored the promise of media-based approaches for enhancing signal quality and transmission efficiency. In particular, the incorporation of Radio Frequency (RF) mirrors—akin to intelligent surfaces and reconfigurable metasurfaces—has emerged as a promising method for improving the reliability and efficiency of wireless links. These mirrors can be strategically positioned to shape and modify the transmitted signal, enabling a favorable signal construction. Following this direction, RF mirrors positioned near or around transmit antennas have led to the development of Media-Based Modulation (MBM) \cite{MBM_khandani_ISIT, MBM_khandani_ISIT2, MBM_Seifi, Basar}. Through careful configuration of the RF mirrors, transmitted signals are \emph{effectively} tuned and enhanced prior to reaching their destination, significantly improving performance.

This paper considers MBM with \emph{multi-state} RF mirrors mounted near a transmit antenna, as opposed to that in \cite{Majid_BMBM} which was confined to the two-state configuration. It is further assumed that the receiver is equipped with a single antenna. We present several new findings. First, we provide an upper bound on the minimum pairwise Euclidean distance of the MBM signal constellation with an arbitrary number of states for single-input and single-output Rayleigh fading links. This bound demonstrates that the average minimum distance decays exponentially with the spectral efficiency of the MBM constellation (i.e., bits per signal point). By comparing this bound with that of the conventional source-based M-QAM modulation, we conclude that there is a loss compared to conventional transmission. The power loss in symbol error rate is estimated to vary from 1 to 2 dB for different spectral efficiencies. Simulation results corroborate these estimates. As a remedy, to improve MBM transmission, we consider \emph{closed-loop} MBM in which feedback is used to \emph{shape} the signal constellations. A numerical algorithm based on a stochastic perturbation method is proposed to construct closed-loop signal constellations with a significantly enhanced minimum distance profile. The simulation results demonstrate that closed-loop MBM \emph{significantly} outperforms open-loop MBM and \emph{approaches} the performance of an AWGN channel. This work not only advances our understanding of MBM but also offers valuable insights for designing future 6G networks, where enhanced signal reliability and efficiency are crucial.

The remainder of the paper is organized as follows. Section~\ref{sec:open_MBM} outlines the baseline open-loop MBM scheme. Section~\ref{sec:open_MBM_properties} discusses some properties of open-loop MBM, such as the minimum distance of MBM signal constellations. Section~\ref{sec:Close_loop MBM} presents the closed-loop MBM approach. Section~\ref{sec:weight_algo} describes a numerical algorithm to optimize the closed-loop MBM. Subsection~\ref{sec:opt_const} provides representative examples of optimized MBM signal constellations. Section~\ref{sec:performance} presents simulation results. Finally, Section~\ref{sec:concl} concludes the paper.



\section{Baseline: Open-Loop MBM}\label{sec:open_MBM}

\begin{figure}[t]
\centering
\vspace{-0.01cm}
\includegraphics[trim=-.75cm 0cm 0cm 0cm, width=.39\textwidth]{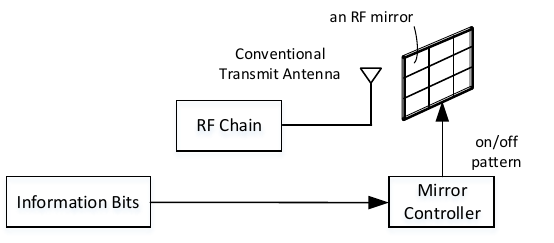}  
\vspace{-0.1cm}
\caption{Baseline  MBM transmitter with RF mirror controller.}
\label{Fig:signaling}
\vspace{-0.3cm}
\end{figure}

Fig.~\ref{Fig:signaling} shows a  block diagram of the transmitter of the open-loop MBM. A single antenna, connected to an RF chain, emits a radio wave at a given frequency. The transmit antenna is surrounded by a set of RF mirrors. The emitted signal will \textquotedblleft pass\textquotedblright~through the mirrors before departing toward its destination. In this figure, nine square-shaped mirrors are shown mounted in front of the transmit antenna. The information bits are passed to a mirror controller, which generates a signal based on the combination of input bits to activate a mirror pattern corresponding to the given information bits.

In \cite{MBM_Seifi}, an MBM prototype is reported in which each pair of adjacent patches can be connected or disconnected using a PIN (Positive-Intrinsic-Negative) diode, i.e. RF switch. 
 An RF mirror will be transparent to the incident wave if its diodes are open, or will reflect the incident wave if its diodes are shorted. Thus, the incoming combination of bits is mapped to the activation patterns of the RF mirrors. 
Consequently, \emph{activation patterns}, each representing a distinct transmit \emph{state}, are formed by combining switched-on and off PIN diodes.

The transmitter can therefore create a set of channel states denoted as $\{s_1, s_2, \ldots, s_{2^k}\}$. For each state, the transmitter creates a channel realization that is mapped to one of $2^k$ binary information strings of length $k$. That is, for each \emph{configured} state, the channel realization \emph{will} change. Let the channel for state $s_i$ be the complex number denoted as  $h_i$. Therefore, the signal constellation points of the MBM are given by the set
\begin{eqnarray}\label{eq:MBM_const}
\mathbb{S} := \Big\{h_1, h_2, \ldots, h_{2^k}\Big\}.
\end{eqnarray}

Trial measurements, such as those reported in \cite{MBM_Seifi}, indicate that the elements of the signal points in $\mathbb{S}$ can be modeled by an i.i.d. Rayleigh fading distribution. This assumption is also considered in \cite{Basar}. We therefore assume throughout that $h_i \sim \mathcal{CN}(0,1)$, i.e., zero-mean unit variance complex Gaussian random variables. Nevertheless, the developed solution is not limited to this distribution of the channel states.

\section{Minimum Distance Properties \\ of Open-Loop MBM}\label{sec:open_MBM_properties}

The minimum pairwise Euclidean distance among the constellation signal points is an important feature of signal constellations in general. In this section, we discuss the minimum distance properties of open-loop MBM constellations. In particular, we show how the minimum distance of the MBM constellations scales as the constellations become denser. We analytically demonstrate that the minimum distance of the open-loop MBM constellation \emph{decreases} at least exponentially with its spectral efficiency.

Let for state $i$, the channel coefficient be $h_i := x_i + \sqrt{-1} y_i$, where $x_i$ and $y_i$ are mutually independent zero-mean complex variables with variance half. Therefore, for two arbitrary points in an MBM constellation, the difference between the constellation points is given by
\begin{eqnarray}\label{eq:BMBM_distance}
    \Delta_{i,j} &=& h_i - h_j \nonumber\\
    &=& (x_i - x_j) + \sqrt{-1} (y_i - y_j) \nonumber\\
    &=:& a_{ij} + \sqrt{-1}b_{ij},
\end{eqnarray}
where $a_{ij}$ and $b_{ij}$ are independent zero-mean Gaussian variables with variance one since $x_i, x_j, y_i$, and $y_j$ are mutually independent with zero mean and variance half. Thus, we have
\begin{eqnarray}\label{eq:BMBM_distance2}
    \Delta_{i,j} \sim \mathcal{CN}(0,2).
\end{eqnarray}

Therefore, the corresponding distance $d_{i,j} := |\Delta_{i,j}|$ is Rayleigh distributed with mean $\sqrt{0.5\pi} \approx 1.25$. The pdf of $d_{i,j}$ is hence given by
\begin{eqnarray}\label{pdf_distance}
    f_{d_{i,j}}(d) = d e^{-\frac{d^2}{2}}.
\end{eqnarray}
That is, an MBM signal constellation with small $d$ is likely if the constellation is drawn randomly.

This situation becomes even more severe if we consider higher-order MBM. To illustrate this, consider a transmitter with $2^k$ states. That is, there are $2^k$ signal points such that each signal point carries $k$ bits. Therefore, the spectral efficiency of the constellation is $k$ bits per signal point. The following proposition sheds some light on this aspect.

\begin{Prop}
For the open-loop M-MBM configured with $M = 2^k$, the average minimum pairwise distance of the signal constellation is bounded as
\begin{eqnarray}\label{eq:MBM_distance_mean}
    \mathbb{E}[d_{min}] &\leq& \sqrt{\pi} 2^{-\frac{k}{2}}.
\end{eqnarray}
\end{Prop}

\begin{proof}
The minimum pairwise distance for a given $2^k$-MBM constellation can be bounded as
\begin{eqnarray}\label{eq:MBM_distance}
    d_{min} &=& \min\left \{ d_{m,n} \mid m \neq n, 1 \leq m \leq 2^k, 1 \leq n \leq 2^k \right\} \nonumber \\
    &\leq& \min \left\{ d_{2i-1,2i} \mid 1 \leq i \leq 2^{k-1} \right\} =: d_o,
\end{eqnarray}
where $d_{m,n} := |h_m - h_n|$ and the inequality follows since we only select a subset of total minimum distances. The reason for this is to create statistically independent minimum pairwise distances to be able to derive analytical expressions.

We next compute the cumulative distribution function of $d_o$ by considering the following series of equalities
\begin{eqnarray}\label{eq:MBM_distance_PDF}
    F_{d_{o}}(d) &=& \text{Pr} (d_{o} \leq d) \nonumber \\
    &=& 1 - \text{Pr} (d_{o} > d) \nonumber \\
    &=& 1 - \text{Pr} \left(\min \{ d_{12}, d_{34}, \ldots, d_{2^k-1,2^k} \} > d\right) \nonumber \\
    &\overset{(a)}{=}& 1 - \text{Pr} \left(d_{12} > d, d_{34} > d, \ldots, d_{2^k-1,2^k} > d\right) \nonumber \\
    &=& 1 - \Big[ \text{Pr} (d_{12} > d) \ldots \text{Pr} (d_{2^k-1,2^k} > d) \Big] \nonumber \\
    &\overset{(b)}{=}& 1 - \left( e^{-\frac{d^2}{2}} \right) \ldots \left( e^{-\frac{d^2}{2}} \right) \nonumber \\
    &=& 1 - e^{-2^{k-2}d^2},
\end{eqnarray}
where (a) holds since $\{d_{12}, d_{34}, \ldots, d_{2^k-1,2^k}\}$ are mutually independent and (b) follows by using \eqref{pdf_distance}.
Thus the pdf of $d_o$ is given by
\begin{eqnarray}\label{eq:MBM_distance_PDF}
    f_{d_{o}}(d) &=& \frac{\partial F_{d_{o}}}{\partial d} = 2^{k-1} d e^{-2^{k-2}d^2}.
\end{eqnarray}
Therefore, $d_o$ is also Rayleigh distributed and we can then obtain the mean value as
\begin{eqnarray}\label{eq:MBM_distance_mean}
    \mathbb{E}[d_{o}] &=& \int_{0}^{\infty} t f_{d_{o}}(t) \, dt \nonumber \\
    &=& \int_{0}^{\infty} 2^{k-1} t^2 e^{-2^{k-2}t^2} \, dt = \sqrt{\pi}2^{-\frac{k}{2}},
\end{eqnarray}
where the last equality can be obtained by recalling the integral equality of \cite{Integral_handbook}
\begin{eqnarray}\label{eq:MBM_distance_mean_integral}
    \int_{0}^{\infty} x^n e^{-a x^2} \, dx = \frac{(2m-1)!!}{2^{m+1}a^m} \sqrt{\frac{\pi}{a}},
\end{eqnarray}
which holds for $n = 2m$ and $a > 0$, where $!!$ denotes the double factorial. This completes the proof.
\end{proof}
The result of the above proposition indicates that the average minimum pairwise distance diminishes \emph{exponentially} with increasing spectral efficiency of the signal constellation.

We can now compare the above result with that of M-QAM. The minimum distance for the received M-QAM signal constellation with unit average symbol energy over Rayleigh fading channels is given by
\begin{eqnarray}\label{eq:mQAM_distance}
    d_{min}^{\text{(M-QAM)}} = \mathbb{E}[|h_i|] \sqrt{\frac{6}{M-1}} = \frac{1}{2} \sqrt{\frac{6 \pi}{2^k-1}},
\end{eqnarray}
where we used the scaled minimum distance of the conventional M-QAM in \cite{Ahlin}.


For instance, for 4-MBM, the average minimum distance of the signal constellation is observed to be at most $d_{min}^{\text{(4-MBM)}} \leq \mathbb{E}[d_{o}] = 0.5 \sqrt{\pi} \approx 0.8862$. In contrast, for the 4-QAM, the minimum distance is given by $d_{min}^{\text{(4-QAM)}} = \sqrt{0.5 \pi} \approx 1.2533$.


In general, for all $k\in \mathbb{N}$, we can write
\begin{eqnarray}\label{eq:mQAM_distance_ratio}
   \eta := \frac{d_{min}^{\text{(M-MBM)}}}{d_{min}^{\text{(M-QAM)}}} \leq \frac{2^{-\frac{k}{2}} \sqrt{\pi}}{\frac{\sqrt{\pi}}{2} \sqrt{\frac{6}{2^k-1}}}= \sqrt{\frac{2}{3}} \left(1 - \frac{1}{2^k}\right)^{1/2} \!\! \! \!< 1,
\end{eqnarray}
which is strictly \emph{less} than one. For  $k\gg 1$, we have
\begin{eqnarray}\label{eq:mQAM_distance_ratio}
    \eta  \lessapprox  \sqrt{\frac{2}{3}} \approx 0.8165.
\end{eqnarray}
Thus, the minimum distance of the MBM constellation is less than that of the corresponding M-QAM constellation, particularly for large $k$.
In conclusion, while the open-loop MBM provides flexibility and (potentially) higher spectral efficiency, it comes with a trade-off in minimum distance compared to conventional M-QAM constellations.

\section{Closed-Loop MBM}\label{sec:Close_loop MBM}
In the previous section, we outlined that the baseline open-loop MBM creates signal constellations with small minimum distances. To address this fundamental issue, we consider a closed-loop MBM, where feedback is used to modify the signal before it reaches the RF mirrors. Fig.~\ref{Fig:closed_loop_MBM} depicts the closed-loop MBM transmitter. The transmitter employs \emph{complex weights} to modify the transmitted signal.

For each channel state from the set $\{s_1, s_2, \ldots, s_{2^k}\}$, the transmitter generates a corresponding set of complex weights given by $\{w_1, w_2, \ldots, w_{2^k}\}$. Thus, the set of closed-loop signal constellation points becomes
\begin{eqnarray}\label{eq:MBM_const}
\mathbb{S}^{\texttt{(cl)}} := \big\{w_1 h_1, w_2 h_2, \ldots, w_{2^k} h_{2^k}\big\}.
\end{eqnarray}
The weights are subject to power constraints. For a unit-power waveform just before weight multiplication, the total transmit power is given by
\begin{eqnarray}\label{eq:CL_MBM_power1}
   \left|w_1 \right|^2 p(s=s_1) + \cdots + \left|w_{2^k} \right|^2 p(s=s_{2^k}).
\end{eqnarray}
Assuming that the input information bits are i.i.d., this creates i.i.d. state changes at the transmitter, such as turning on or off the PIN diodes of the RF mirrors. That is,
\begin{eqnarray}\label{eq:state_probability}
 p(s=s_i) = 2^{-k} \ \text{for} \ i=1,2,\ldots,2^k.
\end{eqnarray}
Since each state is active only once per channel use, the total transmit power for all $2^k$ states must satisfy the constraint
\begin{eqnarray}\label{eq:CL_MBM_power2}
   \left|w_1 \right|^2 + \cdots + \left|w_{2^k} \right|^2 \leq 2^k.
\end{eqnarray}

The power constraint in \eqref{eq:CL_MBM_power2} for the optimal choice should be satisfied by equality. To see this, consider a case where the weights are chosen such that
\begin{eqnarray}\label{eq:CL_MBM_power3}
   \left|w_1 \right|^2 + \cdots + \left|w_{2^k} \right|^2 < 2^k.
\end{eqnarray}
We can then find $\alpha > 1$ such that the new weights satisfy
\begin{eqnarray}\label{eq:CL_MBM_power4}
   \left|\overline{w}_1 \right|^2 + \cdots + \left|\overline{w}_{2^k} \right|^2 = 2^k,
\end{eqnarray}
where $\overline{w}_i = \alpha w_i$ for $1 \leq i \leq 2^k$. The minimum distance of the new constellation is then given by
\begin{eqnarray}\label{eq:MBM_const}
 \overline{d}_{ij} &=& \left|\overline{w}_i h_i - \overline{w}_j h_j \right| = \alpha \left|{w}_i h_i - {w}_j h_j \right| \nonumber\\
  &>& \left|{w}_i h_i - {w}_j h_j \right| = {d}_{ij}.
\end{eqnarray}
Thus, the new constellation has better distance properties since the distance between all pairs is increased.

Therefore, the closed-loop $2^k$-MBM optimization problem can be formulated as
\begin{eqnarray}\label{eq:BMBM}
  && \max \min_{1 \leq i \neq j \leq 2^k} \left|{w}_i h_i - {w}_j h_j \right| \\
 &&  \text{s.t.} \ \ \sum\nolimits_{i=1}^{2^k} \left|w_i \right|^2  = 2^k.
\end{eqnarray}

In the next section, we will discuss a method for solving the formulated MBM optimization problem.

\begin{figure}
\centering
\vspace{-0.0cm}
	\includegraphics[width=.4\textwidth]{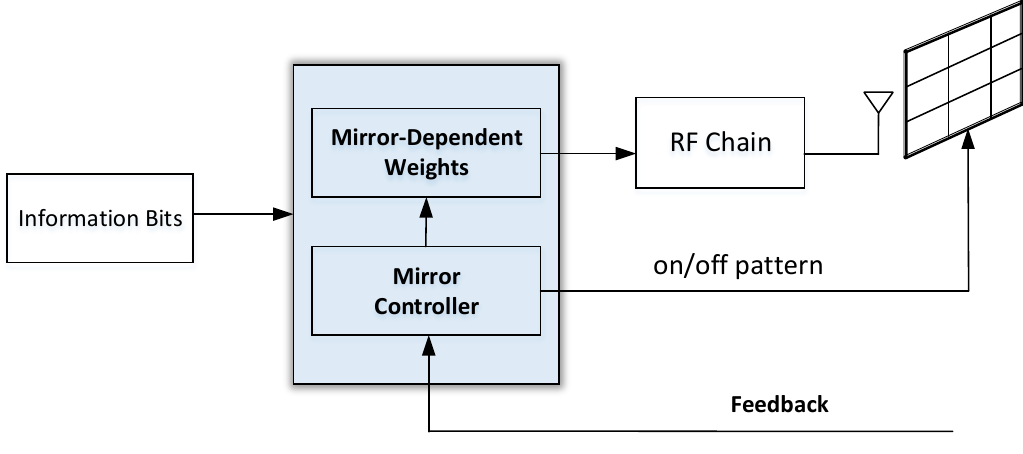}  
\vspace{-0.2cm}
	\caption{Closed-loop MBM transmitter with feedback circuitry.} 
	\label{Fig:closed_loop_MBM}
	\vspace{-0.5cm}
\end{figure}

\begin{figure}
\centering
\vspace{-0.8cm}
	\includegraphics[width=.6\textwidth]{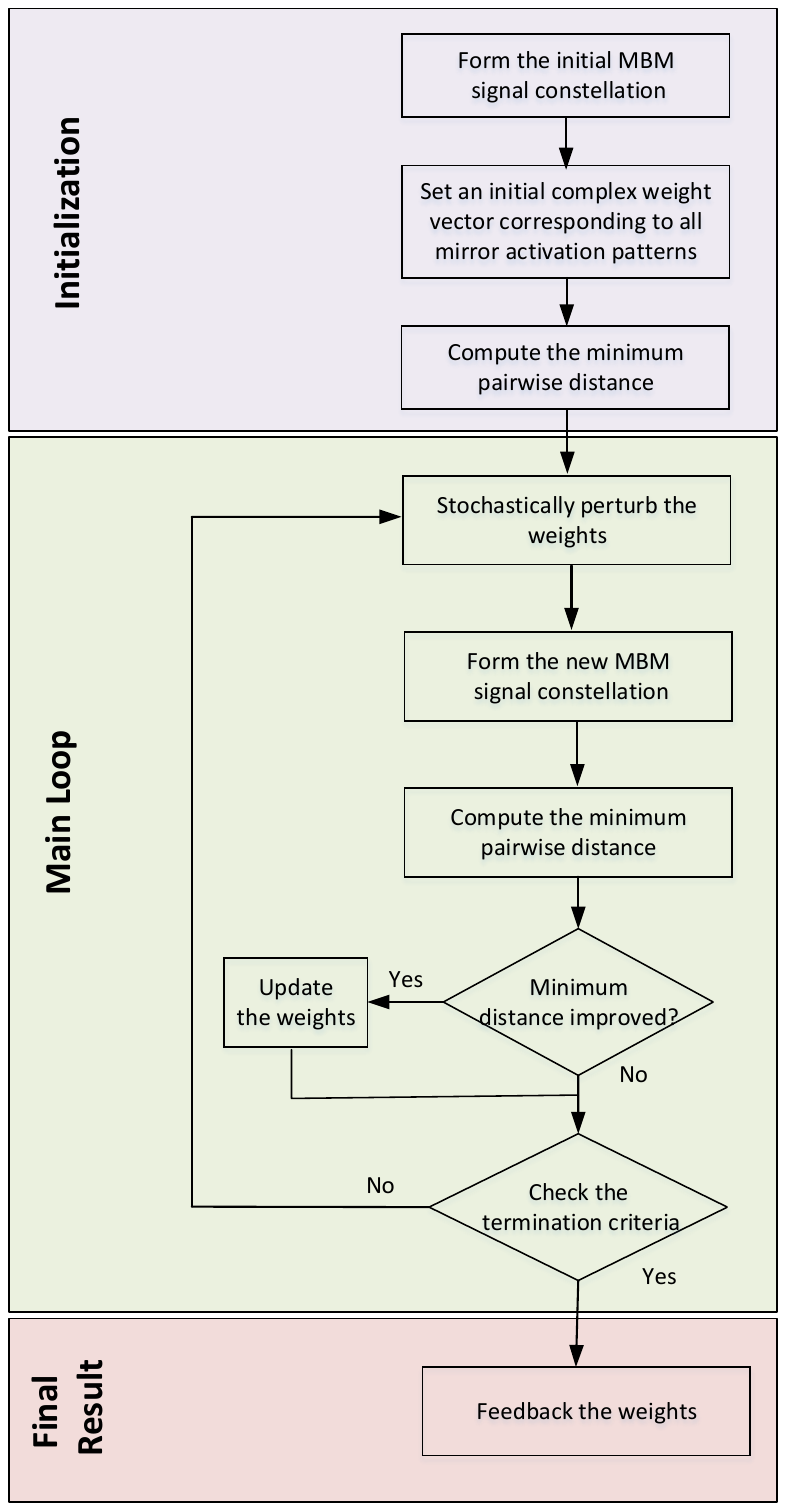}  
\vspace{-1.6cm}
	\caption{Stochastic algorithm for computing optimized weights.}
	\label{Fig:purtebation_algo}
	\vspace{-0.3cm}
\end{figure}

\section{Weight Optimization Algorithm}\label{sec:weight_algo}

Fig.~\ref{Fig:purtebation_algo} illustrates a flowchart of an algorithm designed to compute the complex weights. This algorithm uses an iterative stochastic perturbation method based on a given design metric. It starts with the initial signal space formed by the set $\mathbb{S}:=\left\{h_1, h_2, \ldots, h_{2^k}\right\}$ as its input and produces a set of complex weights $\{w_1, w_2, \ldots, w_{2^k}\}$ as its output, subject to a power constraint. The design metric can include the minimum pairwise distance (e.g., Euclidean distance or any other related distance measure).

To illustrate, consider the minimum pairwise Euclidean distance as the objective function. The algorithm begins with unit weights $w_i^{(0)} = 1$ for all $i=1,2,\ldots,2^k$ (i.e. baseline open-loop configuration) and computes the objective function
\begin{eqnarray}\label{eq:initial_d}
 d_{0} &=& \min_{i \neq j} \left| w_i^{(0)} h_i - w_j^{(0)} h_j \right|.
\end{eqnarray}

Subsequently, the weights are iteratively and \emph{stochastically} perturbed according to
\begin{eqnarray}\label{eq:weight_change}
 w_i^{(l)} &=& \alpha_l \left( w_i^{(l-1)} + \Delta_i^{(l)} \right), \quad \text{for } l=1,2,\ldots
\end{eqnarray}
where $\Delta_i^{(l)}$ denotes the stochastic perturbation variable, selected according to a given random distribution, and $\alpha_l$ is the power normalization factor. At stage $l$, the perturbation $\Delta_i^{(l)}$ is applied to the $i$th weight found in the previous stage $w_i^{(l-1)}$, and the objective function is recomputed:
\begin{eqnarray}\label{eq:perturbed_d}
 d_{l} &=& \min_{i \neq j} \left| w_i^{(l)} h_i - w_j^{(l)} h_j \right|.
\end{eqnarray}

One distribution we use involves first \emph{uniformly} selecting an index from the integer set $\{1,2,\ldots,2^k\}$, and then applying a \emph{uniform circular distribution} to choose the perturbation value $\Delta_i^{(l+1)}$. The support of the uniform distribution may decrease over time, resulting in smaller perturbation values as the algorithm progresses. For unselected indices, no perturbation is applied. The perturbation is accepted if it results in a constellation with a larger minimum distance (i.e., $d_{l} > d_{l-1}$). Otherwise, a new random perturbation is chosen until the design metric is improved. This process continues until there is no significant change in the metric or the maximum number of trials has been reached.

Once the complex weights $w_i$ for $i=1,2,\ldots,2^k$ are determined, these weights (or a signal indicating these weights) are fed back to the transmitter. The transmitter uses these weights to configure its signal. Consequently, the new signal constellation points observed by the receiver are represented by the set $\{\mathtt{w}_i^{\text{(opt)}} h_i\}_{i=1}^{2^k}$, where $\mathtt{W}^{\text{(opt)}}:=\{\mathtt{w}_1^{\text{(opt)}}, \mathtt{w}_2^{\text{(opt)}}, \ldots, \mathtt{w}_{2^k}^{\text{(opt)}}\}$ is the set of optimized complex weights used at the transmitter.

\begin{figure}[t]
\centering
	\includegraphics[width=.45\textwidth]{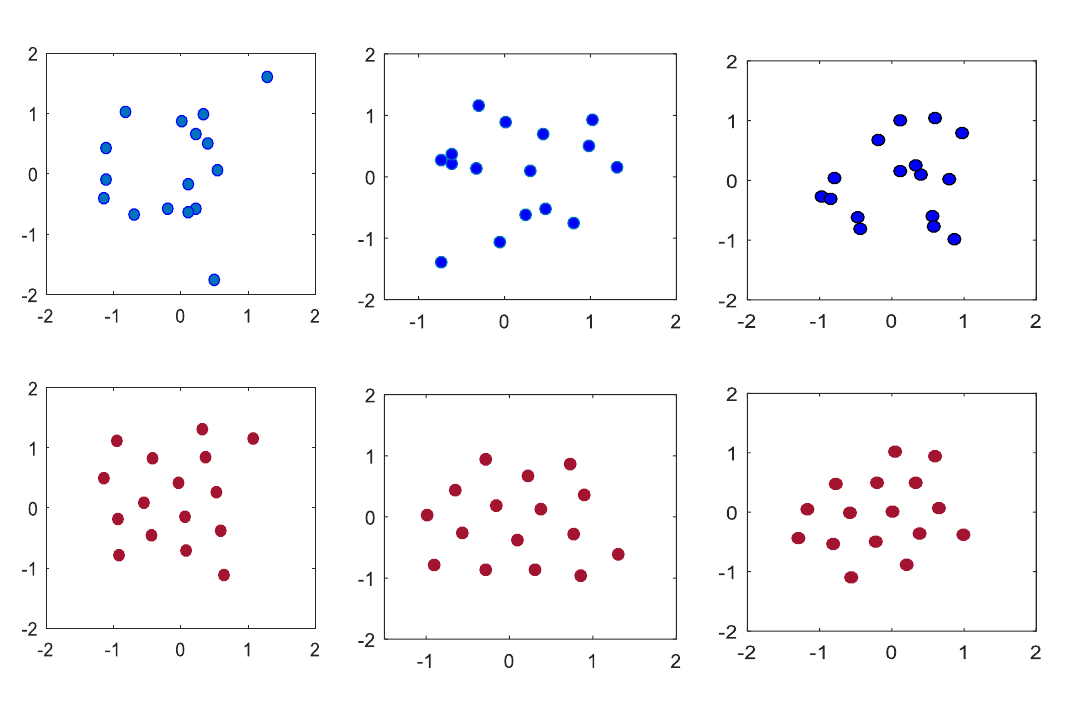}  
\vspace{-0.4cm}
	\caption{Examples of 16-MBM signal point constellations showing both \textcolor{MyRoyalBlue}{\textbf{open-loop}} design and  the corresponding \textcolor{MyFireBrick}{\textbf{optimized closed-loop}} after the configuration of the optimized weights. }
	\label{Fig:Const_16MBM}
	\vspace{-0.4cm}
\end{figure}

\subsection{Exemplary Optimized MBM Constellations}\label{sec:opt_const}

Fig.~\ref{Fig:Const_16MBM} shows the two-dimensional plot of three examples of 16-MBM signal constellations (i.e., RF mirrors with $k=4$ PIN diodes, creating $M=16$ states). The open-loop constellations are generated as described in prior work \cite{Basar} under Rayleigh fading, where the states are activated independently and identically at the transmitter. The top subplots, depicted in blue, show the open-loop signal constellations, while the lower subplots, depicted in red, illustrate the closed-loop signal constellations with optimized weights.


The open-loop constellations are characterized by \emph{low} minimum pairwise distances and irregular shapes, leading to \emph{inefficient} use of the signal space. In contrast, the closed-loop constellations with optimized weights achieve a much larger minimum pairwise distance by adopting shapes similar to \emph{hexagonal} grids, which offer  \emph{optimal packing} of the signal points.  Thus, the algorithm’s ability to enhance constellation regularity from \emph{any} starting  random open-loop state effectively \emph{diminishes} fading effects, leading to closed-loop MBM performance that closely resembles \emph{ideal} AWGN conditions. This results in a \emph{significant} performance improvement.

\begin{figure}[t]
\vspace{-0.4cm}
\centering
	\includegraphics[width=.42\textwidth]{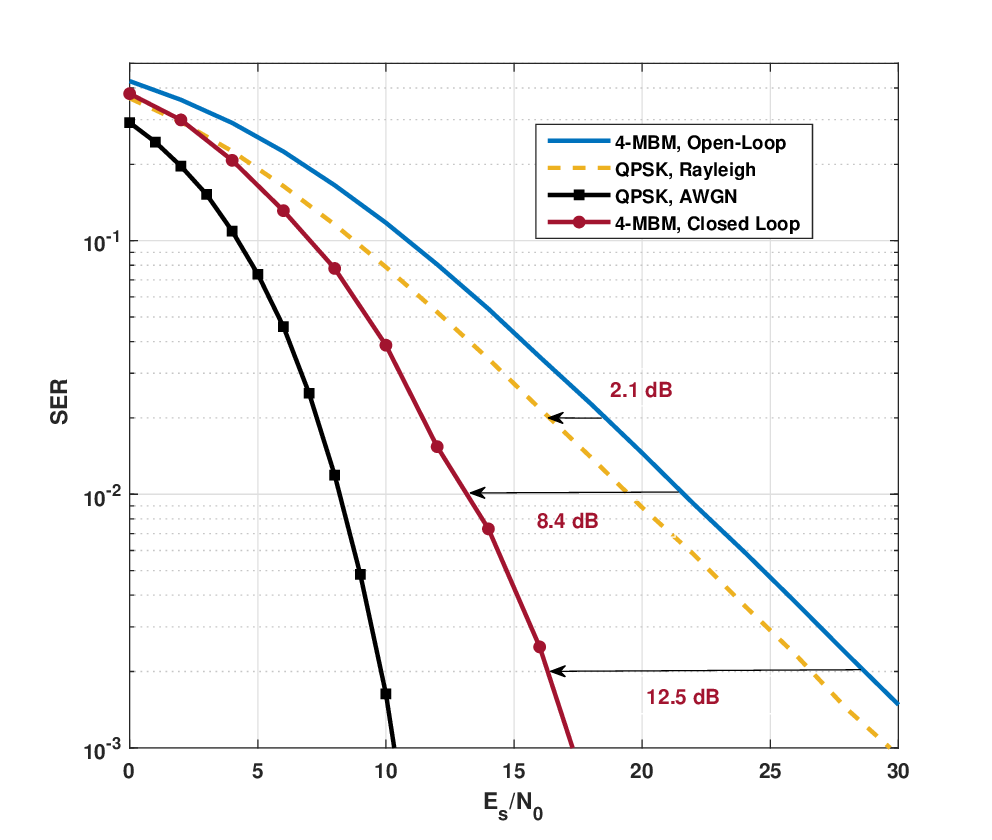}  
\vspace{-0.2cm}
	\caption{SER comparison of 4-MBM with QPSK.}
	\label{Fig:SER_QPSK}
	\vspace{-0.0cm}
\end{figure}

%

\begin{figure}
\centering
\vspace{-0.0cm}
	\includegraphics[width=.42\textwidth]{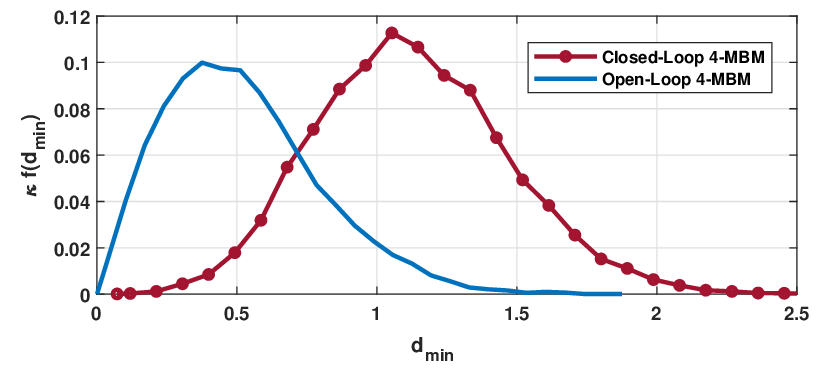}  
\vspace{-0.2cm}
	\caption{Scaled pdf of minimum distance of  MBM.}
	\label{Fig:Distance_QPSK}
	\vspace{-0.4cm}
\end{figure}

\section{Performance Evaluations}\label{sec:performance}

In this section, we discuss the performance of both uncoded and coded MBM transmission schemes for open- and closed-loop configurations and benchmark them against M-QAM performance under Rayleigh fading and AWGN channels.

Fig.~\ref{Fig:SER_QPSK} shows the Symbol Error Rate (SER) for 4-MBM and QPSK schemes. It is observed that the open-loop performs worse than QPSK over Rayleigh fading channels. This power loss, predicted in Section~\ref{sec:open_MBM_properties}, is numerically observed to be about 2 dB. In contrast, the closed-loop with optimized weights performs \emph{significantly} better than that of  the open-loop counterpart, with the gain amounting to several dBs. This notable improvement of the closed-loop 4-MBM over the open-loop can be attributed to the enhanced minimum distance of the signal constellation. Fig.~\ref{Fig:Distance_QPSK} plots the scaled probability density functions (pdfs) of the minimum distance of the constellation for both cases. It is evident that the distribution of $d_{min}$ for the closed-loop is shifted to the right, resulting in a much better minimum distance profile.


Fig.~\ref{Fig:SER_16QAM} illustrates the SER for 16-QAM and both open-loop and closed-loop 16-MBM schemes. The figure shows that the performance degradation of the open-loop system is mitigated in higher-order MBM, as anticipated in Section~\ref{sec:open_MBM_properties}, with a reduction to less than 1 dB for 16-MBM. In \emph{stark} contrast, the closed-loop solution indicates a \emph{substantial} performance improvement over its open-loop counterpart, with its performance approaching that of 16-QAM in an AWGN channel.

\begin{figure}[t]
\vspace{-0.4cm}
\centering
	\includegraphics[width=.42\textwidth]{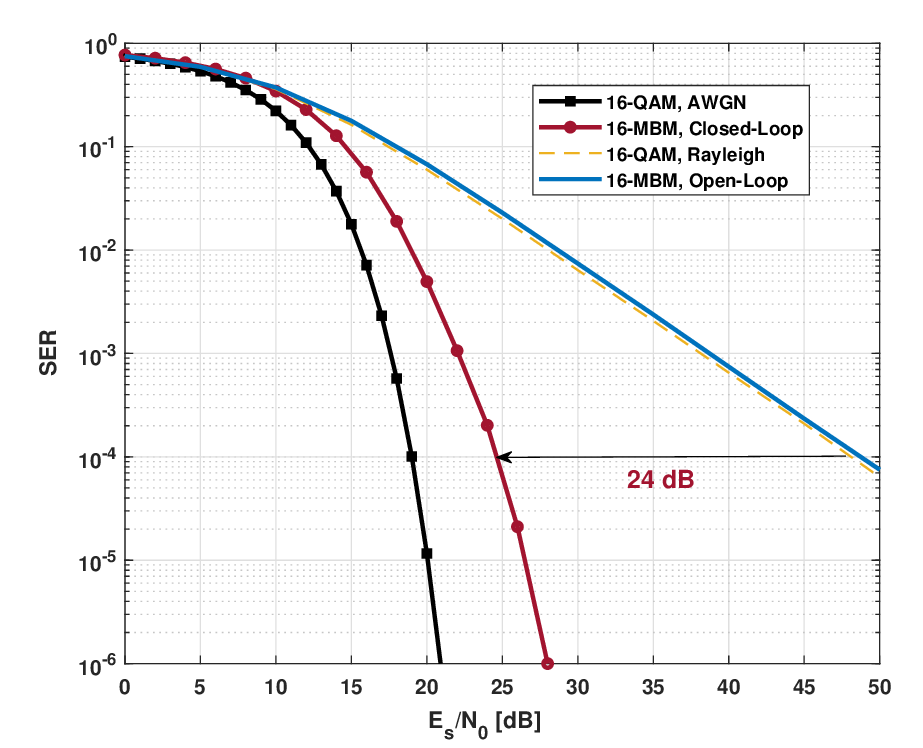}  
\vspace{-0.2cm}
	\caption{SER  Comparison of  16-MBM with 16-QAM.}
	\label{Fig:SER_16QAM}
	\vspace{-0.4cm}
\end{figure}

\begin{figure}[t]
\centering
\vspace{-0.01cm}
	\includegraphics[width=.44\textwidth]{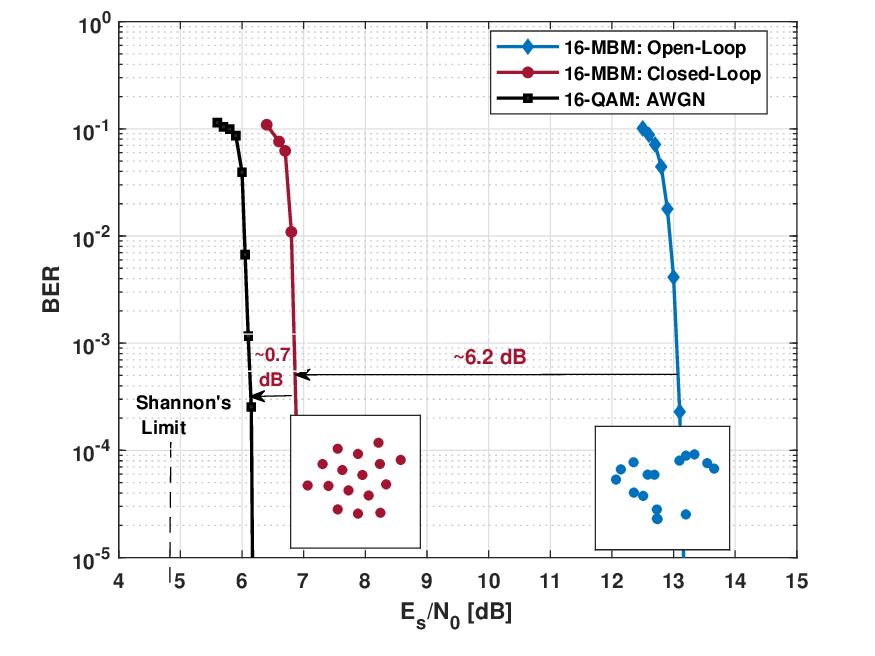}  
\vspace{-0.3cm}
	\caption{BER comparison of  16-MBM with 16-QAM.}
	\label{Fig:BER_16MBM}
	\vspace{-0.2cm}
\end{figure}

Fig.~\ref{Fig:BER_16MBM} depicts the Bit Error Rate (BER) of 16-MBM and 16-QAM over AWGN, where the information bits are rate-half LDPC-coded prior to modulation. For MBM, an optimized symbol-to-bit mapping is used to ensure that the closest points have the minimum Hamming distance. The same stochastic algorithm used in Fig.~\ref{Fig:purtebation_algo} is applied, with the metric changed to Hamming distance and the perturbation implemented through uniform random bit flipping. It is observed that the closed-loop MBM outperforms the open-loop system, achieving performance close to that of AWGN. This improvement is due to the enhanced minimum distance profile achieved through constellation shaping in the closed-loop MBM.


\section{Conclusions}\label{sec:concl}

We introduced a feedback-driven scheme using RF mirrors, where the MBM signal constellation is carefully optimized for desirable features, such as a larger minimum pairwise distance within the signal space. We employed a simple numerical algorithm based on stochastic perturbation and demonstrated a \emph{significant} improvement in performance. Future work may include extending the proposed design with machine learning techniques to potentially enhance and generalize the results.



\end{document}